\documentclass[conference]{IEEEtran}
\IEEEoverridecommandlockouts

\usepackage{microtype}
\usepackage{graphicx}
\usepackage{subfigure}
\usepackage{wrapfig}
\usepackage{booktabs} 

\usepackage{url}
\usepackage{amsmath}
\usepackage{amssymb}
\usepackage{mathtools}
\usepackage{amsthm}
\usepackage{multirow}

\usepackage[capitalize,noabbrev]{cleveref}
\usepackage[capitalize,noabbrev]{cleveref}

\usepackage{cite}
\usepackage{amsmath,amssymb,amsfonts}
\usepackage{algorithmic}
\usepackage{graphicx}
\usepackage{textcomp}
\usepackage{xcolor}
\newtheorem{theorem}{Theorem}[section]

\newtheorem{lemma}[theorem]{Lemma}
\theoremstyle{definition}

\theoremstyle{remark}

\def\BibTeX{{\rm B\kern-.05em{\sc i\kern-.025em b}\kern-.08em
    T\kern-.1667em\lower.7ex\hbox{E}\kern-.125emX}}
\begin{document}
\title{Low Entropy Communication in Multi-Agent Reinforcement Learning}

\author{\IEEEauthorblockN{Lebin Yu, Yunbo Qiu, Qiexiang Wang, Xudong Zhang, Jian Wang}
\IEEEauthorblockA{Department of Electronic Engineering, Tsinghua University. Beijing 100084, China \\
Email: \{yulb19, qyb18, wqx19\}@mails.tsinghua.edu.cn, \{zhangxd, jian-wang\}@tsinghua.edu.cn}
}

\maketitle

\begin{abstract}
Communication in multi-agent reinforcement learning has been drawing attention recently for its significant role in cooperation. However, multi-agent systems may suffer from limitations on communication resources and thus need efficient communication techniques in real-world scenarios. According to the Shannon-Hartley theorem, messages to be transmitted reliably in worse channels require lower entropy. Therefore, we aim to reduce message entropy in multi-agent communication. A fundamental challenge is that the gradients of entropy are either 0 or $\infty$, disabling gradient-based methods. To handle it, we propose a pseudo gradient descent scheme, which reduces entropy by adjusting the distributions of messages wisely. We conduct experiments on two base communication frameworks with six environment settings and find that our scheme can reduce message entropy by up to 90$\%$ with nearly no loss of cooperation performance.
\end{abstract}

\section{Introduction}
Over these years, multi-agent reinforcement learning (MARL) has been attracting increasing attention for its broad applications in cooperative tasks, such as robots navigation \cite{han2020cooperative}, traffic lights control \cite{calvo2018heterogeneous} and large-scale fleet management \cite{lin2018efficient}. To promote the cooperation of agents, a few researchers have designed communication frameworks among agents which adopt neural networks to generate messages \cite{foerster2016learning, ahilan2020correcting, chu2019multi, kim2020communication, wang2021tom2c}.  

However, many multi-agent communication frameworks use cooperation performance as the only metric without considering communication efficiency. Therefore, they generate and send messages at will, making them impractical in scenarios where communication resources are limited \cite{NEURIPS2020_72ab54f9, serra2020whom, sun2020scaling}. To handle this problem, some researchers design efficient multi-agent communication protocols, which can be divided into two categories. The first is decreasing communication times \cite{ma2021learning, kim2018learning, vijay2021minimizing, huang2022importance, hu2021event}, including wisely choosing communication timing and partners. The second is reducing the entropy of generated messages. The motivation is that messages with lower entropy need less bandwidth and thus can be reliably transmitted over worse communication channels according to the Shannon-Hartley Theorem \cite{shannon1948mathematical}. Most learning-based multi-agent communication frameworks use real-valued continuous vectors to communicate, and hence current methods that reduce message entropy all focus on differential entropy\footnote{Following \cite{cover1999elements}, we use discrete entropy to denote the entropy of discrete variables and differential entropy to denote the entropy of continuous variables for ease of reading.} \cite{wang2019learning, wang2020learning, zhang2020succinct}. 

Nevertheless, these methods have three defects. Firstly, some of them rely on specially designed architectures to minimize message entropy \cite{zhang2020succinct}, impairing their generalizability. Secondly, differential entropy is hard to estimate without prior information, and some methods simply treat the message distributions as single Gaussian, which may be far from reality. Thirdly, reducing differential entropy is less significant than reducing discrete entropy of quantized messages. This is because continuous variables must be quantized to discrete variables before being transmitted in a modern communication system, making discrete entropy much more important than differential entropy in efficient communication research.  

In this paper, we propose a scheme, Discrete Entropy Minimization (abbreviated as DisEM), that can be applied to common learning-based multi-agent communication frameworks and reduce the discrete entropy of quantized messages with nearly no cooperation performance decline. The core problem is that quantization truncates gradients and makes all gradient-based training algorithms infeasible. To overcome this challenge, we put forward a novel pseudo gradient descent method that reduces discrete entropy by adjusting the distributions of messages according to well-designed pseudo gradients. We also theoretically prove its effectiveness. An intuitive description of how our DisEM changes the message distribution is that it makes message variables move from less popular quantization intervals to adjacent more popular ones. As a result, DisEM does not change the message distribution too much and hence manages to reduce entropy with little performance degradation. To empirically illuminate DisEM's effectiveness, we conduct experiments in three communication-critical multi-agent tasks with six settings in total. Meanwhile, we apply our scheme to two common multi-agent communication frameworks, IC3NET \cite{singh2018learning} and TARMAC \cite{das2019tarmac}, to manifest its generalizability. Experiments show that DisEM can reduce up to $90\%$ entropy without cooperation performance degradation in some settings. 

\section{System model}

\subsection{Multi-agent reinforcement learning with communication}
We consider a partially observable $n$-agent Markov game \cite{littman1994markov} with communication among agents. This process can be described with a tuple $\langle S, A, R, T, O, \Omega, M_S, n, \gamma \rangle$, where $S$ denotes the state space of the environment, $A$ denotes the set of available actions, $R$ is the reward function $R : S \times A \to \mathbb{R}$, $T$ is the transition function $T: S \times A \to S$, $O$ is the observation space of agents, $\Omega$ is the observation function for agents: $\Omega: S \to O$, $M_S$ denotes the message space, $n$ represents the number of agents, and $\gamma$ is the discount factor.
In a basic MARL framework with communication, an agent $i$ needs a policy $\pi^{i}$ and a message generator $g^{i}$ to finish tasks. At timestep $t$, agent $i$ firstly obtains observation $o^{i}_t$ from the environment and receives messages $\vec{m}^{recv,i}_{t-1}$ from others:
\begin{equation*}
    \vec{m}^{recv,i}_{t-1} = \{\vec{m}^{1}_{t-1}, ..., \vec{m}^{i-1}_{t-1}, \vec{m}^{i+1}_{t-1}, ...,\vec{m}^{n}_{t-1}\}
\end{equation*}
Secondly, it makes decisions $a^{i}_t = \pi^{i}(o^{i}_t, \vec{m}^{recv,i}_{t-1})$ and broadcasts message $\vec{m}^{i}_{t} =g^{i}(o^{i}_t, \vec{m}^{recv,i}_{t-1}) $ to others. Finally, it gets rewards $r^{i}_t$ from the environment. We visualize how agent $i$ interacts with other agents and the environment in Fig.~\ref{fig:agent}.

\begin{figure}[t]
\centering
\includegraphics[width=0.48\textwidth]{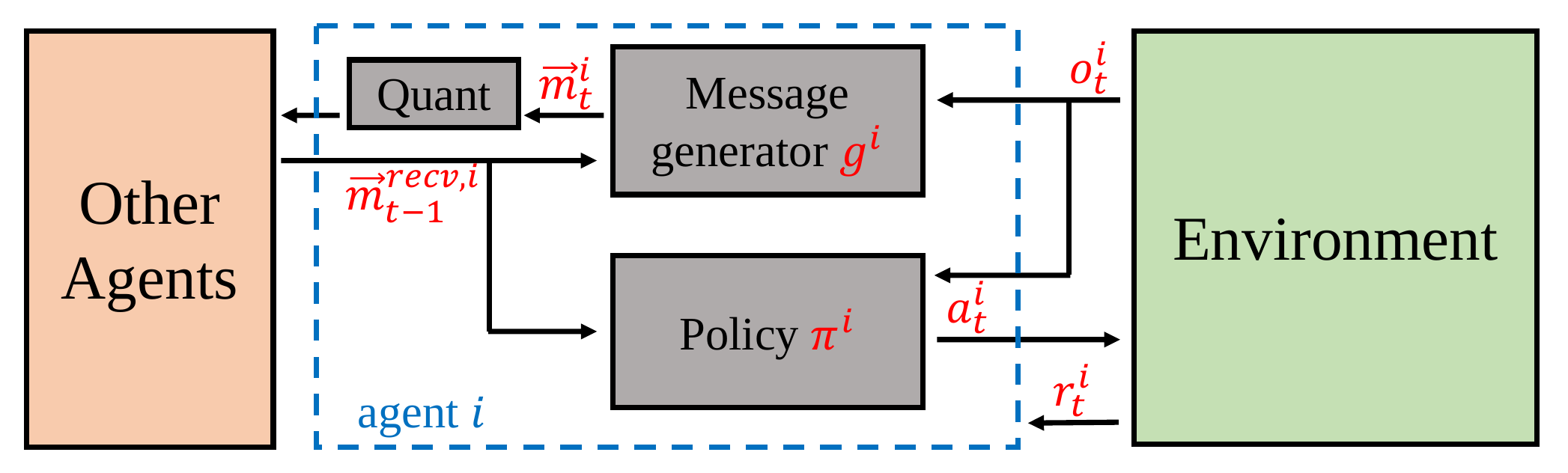}
\caption{How agent $i$ interacts with other agents and the environment.}
\label{fig:agent}
\end{figure}

In our paper, we focus on learning-based multi-agent communication frameworks where $g^{i}$ is implemented through neural networks and trained with gradient-based methods. In this case, we can use $\theta_i$ to denote the parameters of its policy and message generator, and the training goal is to maximize the objective function $J(\theta_i)$:
\begin{equation}
    J(\theta_i) = \mathbb{E}_{\theta_i}\left[ \sum_{t=0}^{\infty} \gamma^t r^{i}_t \right]
\end{equation}
With the help of policy gradient methods \cite{williams1992simple}, $\theta_i$ can be updated based on the gradients of the objective function $\nabla_{\theta_i} J(\theta_i)$.

\subsection{Communication setting}
In our system, messages generated by agents are firstly quantized and then broadcast to others. Specifically, we set the output range of agents' message generators to $[-1,1]$ and use a uniform quantization function $f^Q(x)$ with quantization interval length $\Delta$: 
\begin{equation}
   \label{eqa:quant}
f^Q(x) = k\Delta-1,\ \forall x \in [(k-0.5)\Delta-1 ,(k+0.5)\Delta-1) 
\end{equation}
where $k \in \{0,1,...,K\}$, $K=2/\Delta$ and $K+1$ is the number of quantization intervals. 

In the main experiments, we assume the information transmission is lossless because we aim to reduce the entropy of information sources (i.e. message generators), which is independent of channel conditions. 

\section{Methods}
\subsection{Entropy of quantized messages}
To formulate how entropy of quantized messages is calculated, we define $h_k(\cdot)$ below:
\begin{equation}
   h_k(x) = \left\{
   \begin{aligned}
   1,& \ x\in [(k-0.5)\Delta-1 ,(k+0.5)\Delta-1)\\
   0,& \ x\notin [(k-0.5)\Delta-1 ,(k+0.5)\Delta-1)
   \end{aligned}
   \right.
\end{equation}
Without loss of generality, we firstly focus on the setting where the message length is 1, which means one piece of message is a number instead of a vector. Given a message set $M=\{m_1,m_2,...,m_N\}$, the entropy of quantized messages is calculated below:
\begin{equation}
   H(M) = -\sum_{k=0}^K \frac{\epsilon+\sum_{i=1}^N h_k(m_i)}{N}\log \frac{\epsilon+\sum_{i=1}^N h_k(m_i)}{N}
\end{equation}
where $\epsilon$ is a small number to avoid $\log 0$ in calculation. If the message length is more than $1$, the entropy can be calculated by summing up the entropy of each digit.

\subsection{Reduce entropy without harming cooperation}
We intend to reduce the message entropy of gradient-based multi-agent communication frameworks without harming performance. In specific, for agent $i$, firstly we use the baseline framework to train $\theta_i$ until $\pi^i, g^i$ are well-trained. Secondly, we try to reduce the entropy of messages without harming cooperation performance. A common method is to add a regularizer \cite{tishby2000information, wang2020learning} to the original training objective:
\begin{equation}
   \max_{\theta_i} J'(\theta_i) = J(\theta_i) - \alpha H(M_i)
\end{equation}
where $\alpha$ is the weight of the regularizer and its value depends on the trade-off between maximizing $J(\theta_i)$ and minimizing $H(M_i)$. Then the gradient of $\theta_i$ with respect to this new objective becomes:
\begin{equation}
   \nabla_{\theta_i} J'(\theta_i) = \nabla_{\theta_i}J(\theta_i) - \alpha \nabla_{\theta_i}H(M_i)
\end{equation}

Note that the gradient of $h_k(\cdot)$ is either $0$ or $\infty$. Consequently, $\nabla_{\theta_i}H(M_i)$ is either $0$ or $\infty$. This disables gradient-based training methods and thus makes $H(M_i)$ hard to be reduced. We propose a pseudo gradient descent method to handle this problem.

\begin{figure*}[htbp]
    \centering
    \begin{minipage}[t]{0.44\textwidth}
        \centering
        \includegraphics[width=1\textwidth]{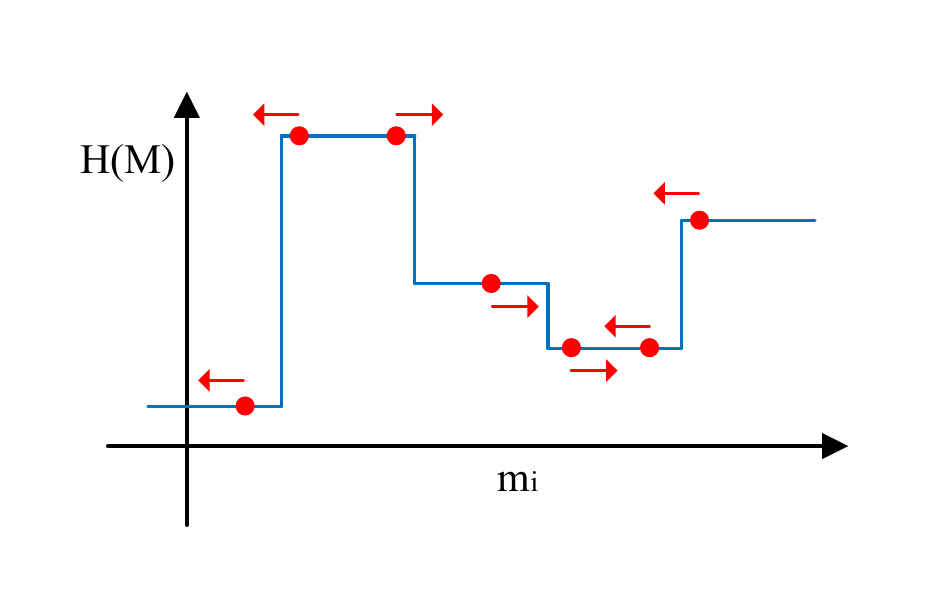}
        \centerline{(a)}
        \end{minipage}
    \begin{minipage}[t]{0.54\textwidth}
        \centering
        \includegraphics[width=1\textwidth]{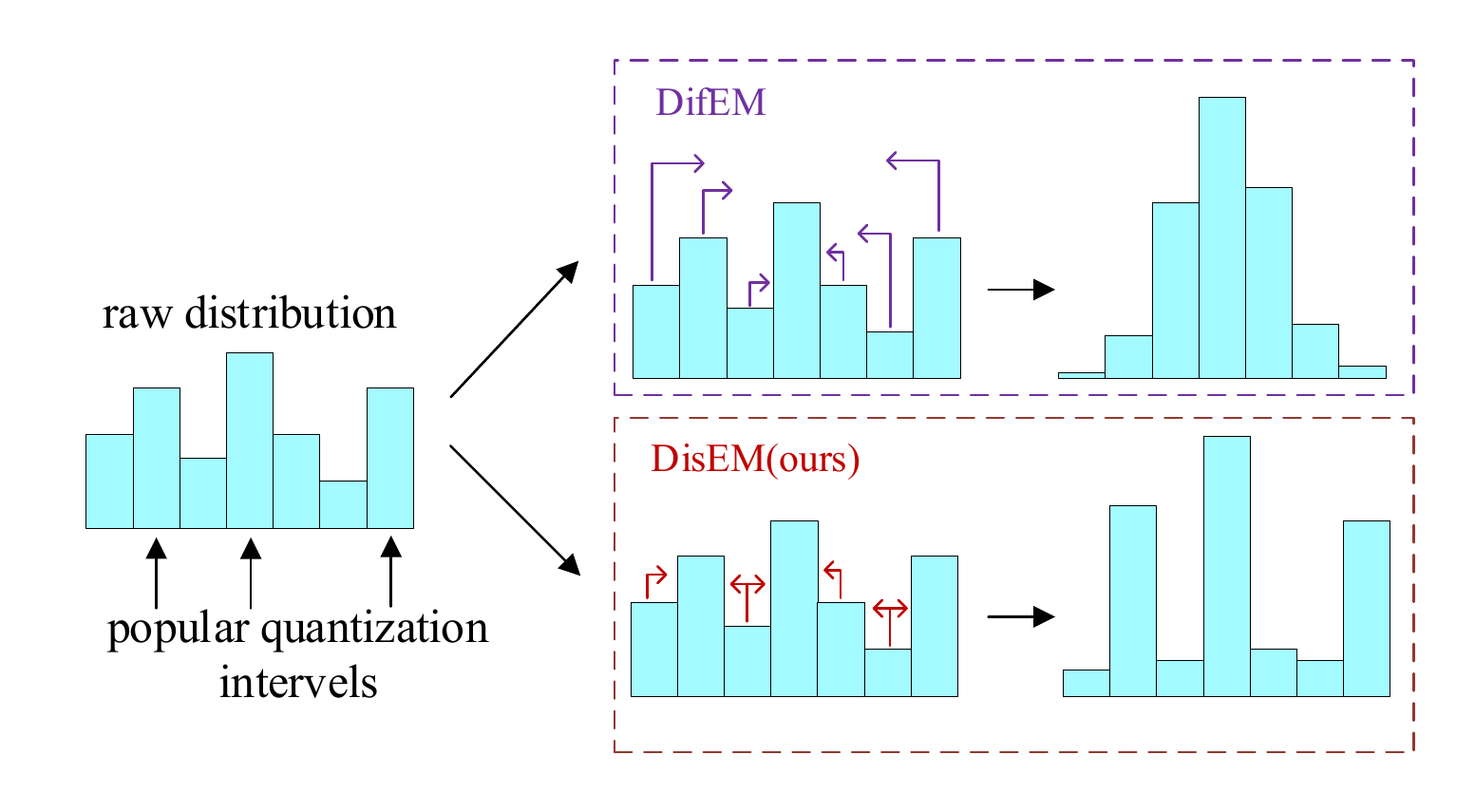}
        \centerline{(b)}
        \end{minipage}
    \caption{(a) How our pseudo gradient descent method reduces $H(M)$ by adjusting $m_i$. The x-coordinate of a red dot represents a possible value for $m_i$, and the y-coordinate of it represents the corresponding value for $H(M)$. Red arrows indicate directions of pseudo gradient descent. Pseudo gradient descent makes $m_i$ move to a direction that might reduce $H(M)$. (b) How our scheme (Discrete Entropy Minimization, abbreviated as DisEM) and traditional scheme (Differential Entropy Minimization, abbreviated as DifEM) change the distributions of messages. The y-axis represents probability and x-axis represents the value of message variables. DisEM reduces entropy by moving values of messages from less popular quantization intervals to adjacent more popular ones, while DifEM simply makes the distribution more like a low variance Gaussian distribution.}
    \label{fig:showre}
\end{figure*}

\subsection{Reduce entropy with pseudo gradient descent}
$H(M)$ can be treated as a multivariate function with $N$ variables $H(m_1,m_2,...,m_N)$, and in this part we focus on how to reduce $H(M)$ by adjusting $m_i$. To start with, we present the core function of gradient descent methods: given a continuously differentiable function $f(x_1,x_2,...,x_n)$ and $\eta \to 0$,

\begin{equation}
    \label{eqa:gd}
    \begin{split}
        x'_i = x_i& - \eta \nabla_{x_i}f \\
        f(x_1,...,x'_i,...,f_n) &\leq f(x_1,...,x_i,...,f_n)
    \end{split}
\end{equation}

Since $\nabla_{m_i}H(M)= \sum_k\nabla_{h_k}H(M)\nabla_{m_i}h_k =0$  or $\infty$, we cannot use gradient descent to minimize $H(M)$. Therefore, we try to design a pseudo gradient $\nabla^p_{m_i}H(M)$ to achieve a similar effect to (\ref{eqa:gd}). 

Our core idea is to replace $\nabla_{m_i}h_k$ with $s_k(m_i)$:
\begin{equation}
    s_k(x) = \left\{
    \begin{aligned}
    1,& \ x\in ((k-1)\Delta-1 ,k\Delta-1) \\
    -1,& \ x\in (k\Delta-1 ,(k+1)\Delta-1) \\
    0, & \ x=k\Delta-1
    \end{aligned}
    \right.
 \end{equation}
and the expression for pseudo gradient $\nabla^p_{m_i}H(M)$ is :
\begin{equation}
    \nabla^p_{m_i}H(M) = \sum_k\nabla_{h_k}H(M)s_k(m_i)
\end{equation}
Next, we show that pseudo gradient descent has a similar property to gradient descent. Given $H(M)=H(m_1,...,m_i,...,m_n)$ and $\eta \to 0$,
\begin{equation}
    \label{eqa:core}
    \begin{split}
        m'_i = m_i &- \eta \nabla^p_{m_i}H(M) \\
        H(m_1,...,m'_i,...,m_N) &\leq H(m_1,...,m_i,...,m_n)
    \end{split}
\end{equation}
The proof is presented in the next subsection. 

We visualize how our pseudo gradient descent method adjusts $m_i$ in Fig.~\ref{fig:showre}(a) and how our scheme and previous schemes (those that aim to minimize differential entropy) change the distribution of messages in Fig.~\ref{fig:showre}(b). In brief, our scheme reduces entropy by moving the values of messages from less popular quantization intervals to adjacent more popular ones and does not change the messages too much. As a comparison, traditional schemes simply make the distribution more like a low variance Gaussian distribution regardless of what the original distribution is.
\begin{figure*}[htbp]
    \centering
    \begin{minipage}[t]{0.32\textwidth}
    \centering
    \includegraphics[width=1\textwidth]{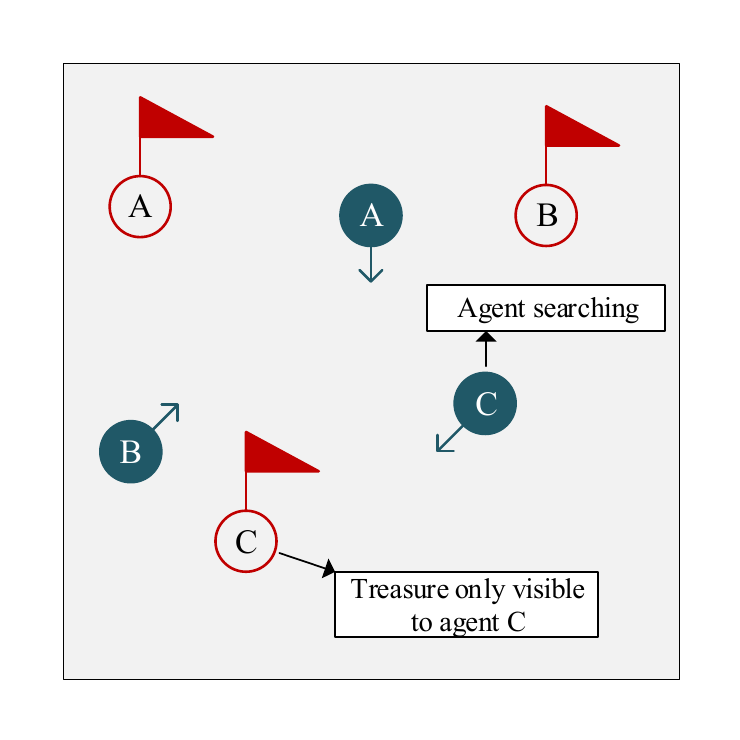}
    \end{minipage}
    \begin{minipage}[t]{0.32\textwidth}
    \centering
    \includegraphics[width=1\textwidth]{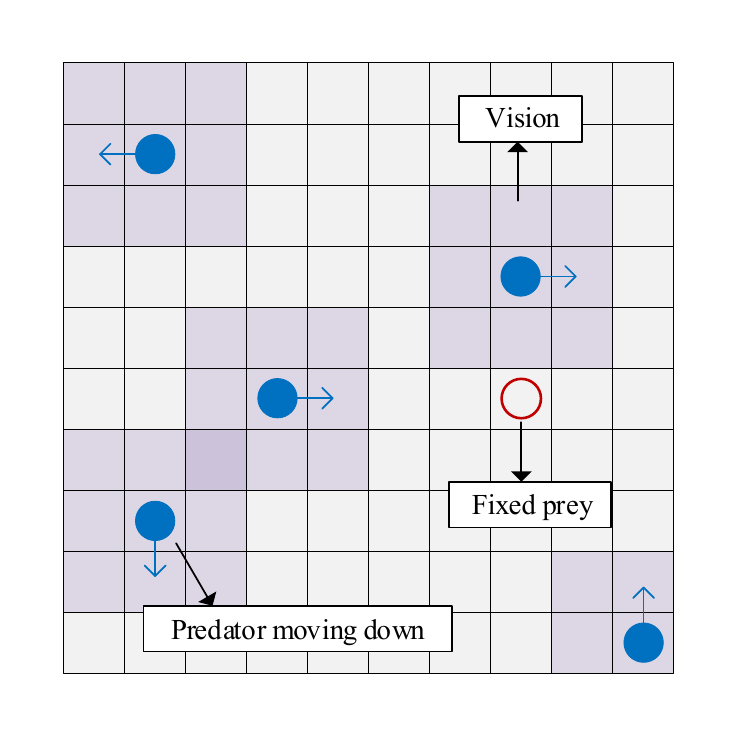}
    \end{minipage}
    \begin{minipage}[t]{0.32\textwidth}
    \centering
    \includegraphics[width=1\textwidth]{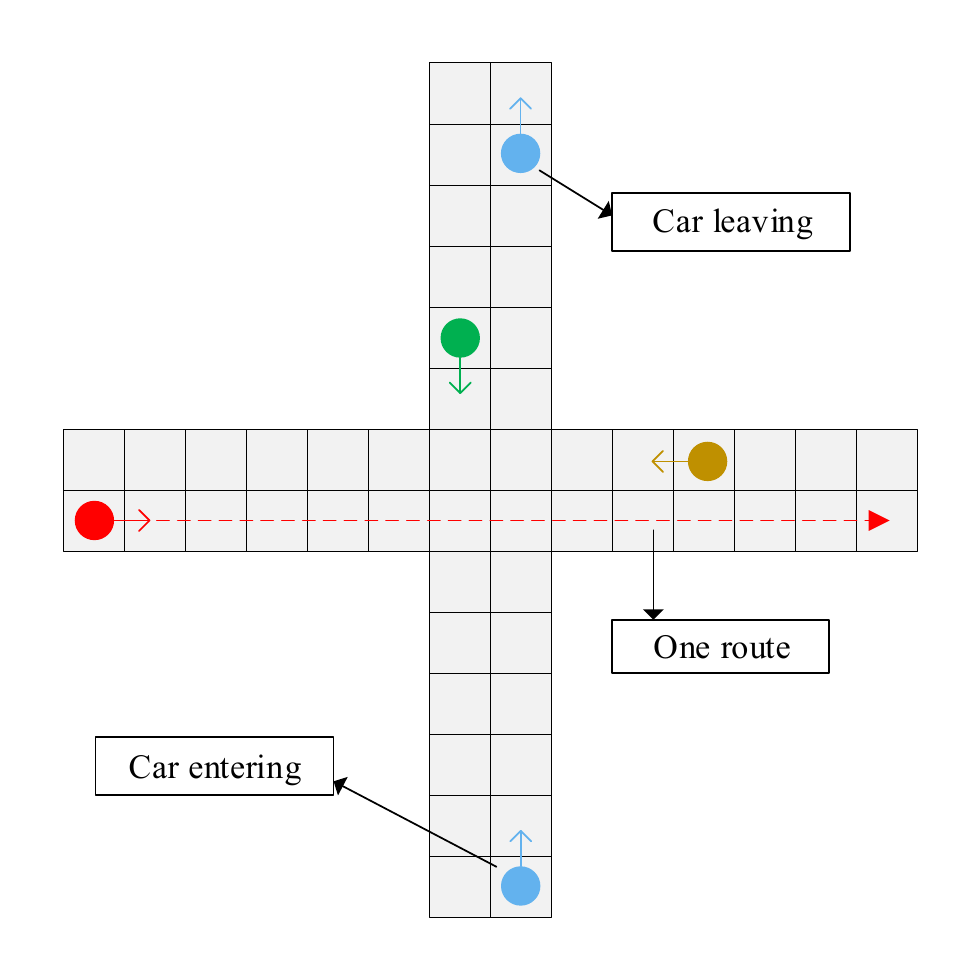}
    \end{minipage}
    \caption{Visualizations of our three experiment environments: Treasure Hunt, Predator Prey, and Traffic Junction.}
    \label{fig:env}
    \end{figure*}
\subsection{Proof of (\ref{eqa:core})}
We set $N_k = \epsilon+\sum_{i=1}^N h_k(m_i)$ for brevity. Note that $N_k$ represents the number of message variables that fall into the quantization interval $[(k-0.5)\Delta-1, (k+0.5)\Delta-1)$. To prove (\ref{eqa:core}), we present three lemmas:
\begin{lemma}
    \label{lemma:1}
    If $m_i \in (u\Delta-1, (u+1)\Delta-1)$, then
    \begin{equation}
        \begin{split}
            \nabla^p_{m_i}H(M) &> 0, \ if \ N_u > N_{u+1} \\
            \nabla^p_{m_i}H(M) &< 0, \ if \ N_u < N_{u+1} \\
            \nabla^p_{m_i}H(M) &= 0, \ if \ N_u = N_{u+1}
        \end{split}
    \end{equation}
\end{lemma}
\textbf{Remark} This lemma shows that for two adjacent quantization intervals, pseudo gradient descent method will make message variables move from the less popular one to the more popular one.
\begin{proof}
\begin{equation}
    \begin{split}
    \nabla^p_{m_i}H(M) &= \sum_{k=0}^K\nabla_{h_k}H(M)s_k(m_i) \\
                    &= -\frac{1}{N}\sum_{k=0}^K(1+\log\frac{N_k}{N})s_k(m_i)
\end{split}
\end{equation}
Note that $\forall m_i \in (u\Delta-1, (u+1)\Delta-1)$, $s_u(m_i) = -1$, $s_{u+1}(m_i) = 1$ and $s_k(m_i) = 0 \ \forall k \in \{0,1,...,K\}\setminus \{u,u+1\}$. Therefore, we get
\begin{equation}
    \nabla^p_{m_i}H(M) = -\frac{1}{N} log \frac{N_{u+1}}{N_u}
\end{equation}
Then we can lead to the equations in Lemma~\ref{lemma:1}
\end{proof}
\begin{lemma}
    \label{lemma:2}
    If $N_u>N_{u+1}$ and $m_i \in (u\Delta-1, (u+1)\Delta-1)$ is updated to $m'_i$ with $m'_i = m_i -\eta \nabla^p_{m_i}H(M), \eta \to 0\ $, this update leads to only two possible results:

    (a) $N'_k = N_k, \ \forall k \in\{ 0,1,...,K \}$

    (b) $N'_u = N_u+1$, $N'_{u+1} = N_{u+1}-1$ and $N'_k = N_k, \ \forall k \in\{ 0,1,...,K \} \setminus \{u,u+1\}$
    
    If $N_u<N_{u+1}$, similarly, this update leads to only two possible results:

    (a) $N'_k = N_k, \ \forall k \in\{ 0,1,...,K \}$

    (b) $N'_u = N_u-1$, $N'_{u+1} = N_{u+1}+1$ and $N'_k = N_k, \ \forall k \in\{ 0,1,...,K \} \setminus \{u,u+1\}$
\end{lemma}
\begin{proof}
    We firstly focus on the first case (i.e. $N_u>N_{u+1}$). Since $\nabla^p_{m_i}H(M)>0$, 
    \begin{equation}
        m'_i = m_i -\eta \nabla^p_{m_i}H(M) \in (m_i-\eta\max(\nabla^p_{m_i}H(M)), \  m_i)
    \end{equation}
    where 
    \begin{equation}
        \max(\nabla^p_{m_i}H(M)) = \frac{1}{N}\log\frac{N_s+\epsilon}{\epsilon}
    \end{equation}
    Since $\eta \to 0$, $\eta\max(\nabla^p_{m_i}H(M))<0.5\Delta$. Note that $m_i \in (u\Delta-1, (u+1)\Delta-1)$, so $m'_i \in ((u-0.5)\Delta-1, (u+1)\Delta-1)$. Consequently, there are only two possibilities for the values of $m_i$ and $m'_i$:

    (a) If $m_i \in ((u+0.5)\Delta-1, (u+1)\Delta-1)$ and $m'_i \in ((u-0.5)\Delta-1, (u+0.5)\Delta-1)$, then $N'_u = N_u+1$, $N'_{u+1} = N_{u+1}-1$ and $N'_k = N_k, \ \forall k \in\{ 0,1,...,K \} \setminus \{u,u+1\}$ 

    (b) Otherwise, $N'_k = N_k, \ \forall k \in\{ 0,1,...,K \}$

    The proof is similar in second case (i.e. $N_u<N_{u+1}$)
\end{proof}

Using $H(M)$ to denote $H(m_1,...,m_i,...,m_N)$ and $H(M')$ to denote $(m_1,...,m'_i,...,m_N)$, we propose our third lemma:
\begin{lemma}
    \label{lemma:3}
    $H(M')<H(M)$ if $N_u>N_{u+1}$, $N'_u = N_u+1$, $N'_{u+1} = N_{u+1}-1$ and $N'_k = N_k, \ \forall k \in\{ 0,1,...,K \} \setminus \{u,u+1\}$.
    
    $H(M')<H(M)$ if $N_u<N_{u+1}$, $N'_u = N_u-1$, $N'_{u+1} = N_{u+1}+1$ and $N'_k = N_k, \ \forall k \in\{ 0,1,...,K \} \setminus \{u,u+1\}$.    
\end{lemma}
\begin{proof}
    We firstly focus on the first case (i.e. $N_u>N_{u+1}$).

    \begin{equation}
    \begin{split}
        H(M')-H(M) = \frac{N_u}{N}\log\frac{N_u}{N}+\frac{N_{u+1}}{N}\log\frac{N_{u+1}}{N}\\ -  \frac{N_u+1}{N}\log\frac{N_u+1}{N}-\frac{N_{u+1}-1}{N}\log\frac{N_{u+1}-1}{N}
    \end{split}
    \end{equation}
    For brevity, we set $x_1 = \frac{N_u}{N}$, $x_2 = \frac{N_{u+1}}{N}$, $\delta=\frac{1}{N}$ and $f(x)=x\log x$. Then we get:
    \begin{equation}
        H(M')-H(M) = f(x_2)-f(x_2-\delta) - (f(x_1+\delta)-f(x_1))
    \end{equation}
    According to the mean value theorem, there exists $x_{1m} \in (x_1,x_1+\delta)$ and $x_{2m}\in(x_2-\delta,x_2)$ such that $f(x_2)-f(x_2-\delta) = \delta f'(x_{2m})$ and $f(x_1+\delta)-f(x_1) = \delta f'(x_{1m})$. Since $f''(x)>0$ and $x_{2m}<x_2<x_1<x_{1m}$, $H(M')-H(M)<0$. The proof is similar in the second case.
\end{proof}
From Lemma~\ref{lemma:2} and Lemma~\ref{lemma:3} we conclude that, if $H(M')\neq H(M)$, $H(M')<H(M)$. Therefore, we obtain $H(M') \leq H(M)$.

\subsection{Implementation details}
Following policy gradient methods and our proposed pseudo gradient descent scheme, the parameters $\theta_i$ of agent $i$ can be trained based on the following gradient expression:
\begin{equation}
    \nabla_{\theta_i} J'(\theta_i) = \nabla_{\theta_i}J(\theta_i) - \alpha \nabla^p_{\theta_i}H(M_i)
 \end{equation}
Besides, we first train the agents without the entropy regularizer (i.e. set $ \alpha=0 $) to make them optimize their policies and message generators. After $T_{N}$ epochs, we add the regularizer (i.e. set $ \alpha=\alpha_p $) and continue to train them for another $T_{max}-T_{N}$ epochs. $T_N, \alpha_p$ and $T_{max}$ are predefined hyperparameters, whose values depend on specific tasks. As for the quantization function, we set $\Delta=0.25$.

\section{Experiments}

\begin{table*}[htbp]
  \caption{Cooperation performance and message entropy of eight communication frameworks in six settings. }
\label{table:mainresults}
\begin{center}
\begin{small}
\begin{sc}
\resizebox{\textwidth}{!}{
 \begin{tabular}{cl|cccc|cccc}
 \hline
   &   & \multicolumn{4}{c|}{IC3NET-based} & \multicolumn{4}{c}{TARMAC-based} \\
   &   & ORI & ZC & DifEM & DisEM & ORI & ZC & DifEM & DisEM \\
 \hline
 \multirow{2}[2]{*}{TH-A} & timesteps↓ & 8.3 $\pm$0.3 & 19.7 $\pm$0.1 & 9.1 $\pm$0.4 & 9.0 $\pm$0.3 & 8.2 $\pm$0.2 & 19.7 $\pm$0.2 & 10.0 $\pm$0.5 & 9.2 $\pm$0.2 \\
   & entropy↓ & 139 $\pm$7 & 0 & 70 $\pm$2.8 & 15 $\pm$1 & 70 $\pm$1 & 0 & 29 $\pm$1 & 9 $\pm$1 \\
 \hline
 \multirow{2}[2]{*}{TH-B} & timesteps↓ & 24.4 $\pm$1.0 & 59.9 $\pm$0.0 & 42.9 $\pm$6.5 & 24.1 $\pm$0.7 & 21.0 $\pm$0.6 & 59.9 $\pm$0.0 & 36.3 $\pm$6.0 & 22.8 $\pm$0.4 \\
   & entropy↓ & 127 $\pm$4 & 0 & 39 $\pm$26 & 30 $\pm$1 & 72 $\pm$4 & 0 & 25 $\pm$6 & 15 $\pm$1  \\
 \hline
 \multirow{2}[2]{*}{PP-A} & timesteps↓& 10.1 $\pm$0.1 & 17.4 $\pm$0.5 & 14.4 $\pm$2.5 & 10.7 $\pm$0.4 & 10.0 $\pm$0.2 & 17.5 $\pm$0.2 & 11.3 $\pm$0.4 & 10.4 $\pm$0.1  \\
   & entropy↓ & 94 $\pm$3 & 0 & 20 $\pm$14 & 9 $\pm$1 & 60 $\pm$3 & 0 & 13 $\pm$1 & 6 $\pm$1  \\
 \hline
 \multirow{2}[2]{*}{PP-B} & timesteps↓ & 16.5 $\pm$1.4 & 33.1 $\pm$0.7 & 17.5 $\pm$0.5 & 16.3 $\pm$0.7 & 16.0 $\pm$1.6 & 33.6 $\pm$0.5 & 16.8 $\pm$0.5 & 15.4 $\pm$0.5  \\
   & entropy↓ & 122 $\pm$11 & 0 & 67 $\pm$6 & 27 $\pm$3 & 58 $\pm$5 & 0 & 25 $\pm$3 & 14 $\pm$1  \\
 \hline
 \multirow{2}[2]{*}{TJ-A} & success rates↑ & 0.87 $\pm$0.13 & 0.29 $\pm$0.02 & 0.66 $\pm$0.05 & 0.85 $\pm$0.11 & 0.77 $\pm$0.18 & 0.28 $\pm$0.02 & 0.71 $\pm$0.04 & 0.74 $\pm$0.21  \\
   & entropy↓ & 141 $\pm$10 & 0 & 96 $\pm$26 & 83 $\pm$10 & 27 $\pm$4 & 0 & 27 $\pm$13 & 27 $\pm$4  \\
 \hline
 \multirow{2}[2]{*}{TJ-B} & success rates↑ & 0.95 $\pm$0.02 & 0.74 $\pm$0.02 & 0.73 $\pm$0.02 & 0.93 $\pm$0.01 & 0.95 $\pm$0.02 & 0.72 $\pm$0.02 & 0.90 $\pm$0.09 & 0.95 $\pm$0.01  \\
   & entropy↓ & 75 $\pm$30 & 0 & 77 $\pm$32 & 38 $\pm$29 & 44 $\pm$2 & 0 & 22 $\pm$10 & 14 $\pm$6  \\
 \hline
 \end{tabular}%

}

\end{sc}
\end{small}
\end{center}
\end{table*}%

\subsection{Environments}
We consider three environments, each with two settings, for demonstration purposes: Treasure Hunt (TH), Predator Prey (PP), and Traffic Junction (TJ)\cite{singh2018learning}, which are visualized in Fig.~\ref{fig:env}. They are all communication-critical, where reducing message entropy is meaningful.  

\textbf{Treasure Hunt} In this task, $N$ agents work together to hunt treasures in the field. Each agent obtains the coordinates of its treasure, which is invisible to others. Note that an agent cannot collect its treasure by itself. Instead, it should help others hunt it through communication. The field size is $1\times1$, and an agent can move in eight directions at speed $v$. One episode ends if all treasures are found or the timestep reaches the upper limit $t_{max}$. Therefore, smaller timesteps indicate better performance. In setting A (TH-A), $N=3, v=0.15$ and $t_{max}=20$. In setting B (TH-B), $N=6, v=0.09$ and $t_{max}=60$. As for the training hyperparameters, $T_N=100$, $\alpha_p=0.2$ and $T_{max}=200$.

\textbf{Predator Prey} \cite{singh2018learning} In this task, $N$ agents with limited vision are required to reach a fixed prey in a grid world of size $D \times D$. One episode ends if all agents reach the prey or the timestep reaches the upper limit $t_{max}$. Therefore, smaller timesteps indicate better performance. In setting A (PP-A), $N=3$, $D=5$, $t_{max}=20$ and the vision is set to $0$. In setting B (PP-B), $N=5$, $D=10$, $t_{max}=40$ and the vision is set to $1$.  As for the training hyperparameters, $T_N=100$, $\alpha_p=0.05$ and $T_{max}=150$.

\textbf{Traffic Junction} \cite{sukhbaatar2016learning} In this task, cars enter a junction from all entry points with a probability $p_{arr}$. The maximum number of cars present is set to $N$. Each car is assigned a fixed route, and they have two options at each step: move along the route or stay. Cars are required to finish routes quickly without collisions. If no collision happens after $t_{max}$, this episode is counted as a success, or else a failure. It is worth noticing that a car only observes its location and must communicate to learn the locations of other cars, thereby avoiding collisions. In setting A (TJ-A), $N=5$, $p_{arr} = 0.3$ and $t_{max}=20$. In setting B (TJ-B), $N=10$, $p_{arr} = 0.05$ and $t_{max}=40$. As for the training hyperparameters, $T_N=1000$, $\alpha_p=0.05$ and $T_{max}=1250$.

\subsection{Tested algorithms}

To test the effectiveness of our scheme, the tested algorithms are based on two  multi-agent communication frameworks: IC3NET \cite{singh2018learning} and TARMAC \cite{das2019tarmac}. In particular, we test four variants of each framework as described below. 

\textbf{Original(ORI)}: This is the original version of the framework without any modifications, which manifests baseline cooperation performance and entropy. 

\textbf{ZeroComm(ZC)}: This variant disables communication by compulsively setting all messages to zero and manifests performance without communication. 

\textbf{DifEM}: \cite{wang2020learning} A previous method that minimizes the differential entropy with a mutual information regularizer. 

\textbf{DisEM}: This is our proposed scheme which reduces entropy by adding a pseudo entropy regularizer to the original loss function.

\subsection{Main experiments results}
We conduct main experiments in these six settings with eight methods in total.  To enhance the reliability of the results, we repeatedly train each setting using each method five times and test each model for 500 episodes. Then we exhibit the averaged outcomes in Table~\ref{table:mainresults}. Here we stress that our DisEM scheme aim to achieve good cooperation performance while keeping the entropy low\footnote{TARMAC-based methods generate messages with less entropy than IC3NET-based ones, which results from the different settings in original papers. Default message size in IC3NET is 128, while in TARMAC it is 48.}. Note that in TH and PP scenarios, lower timesteps indicate better cooperation, while in TJ scenarios, higher success rates indicate better cooperation. 

We conclude several facts from the results. (a) ZC (ZeroComm) misbehaves in almost all settings, especially in TH-A and TH-B, where agents are entirely unable to complete tasks without communication. This reflects the importance of communication in these settings. (b) DifEM indeed reduces communication entropy at the cost of degrading performance more or less. (c) DisEM outperforms DifEM, managing to reduce entropy with little or no performance degradation. More precisely, it works pretty well in TH and PP settings, reducing entropy by $80\%-90\%$. However, in TJ-A setting, both TARMAC-DifEM and TARMAC-DisEM fail to reduce entropy. This is because the message distribution in this setting is hard to optimize. 

\subsection{Investigative experiment: reducing message length}
Reducing layer size is an efficient and common way to compress information, as is used in autoencoders \cite{kramer1991nonlinear}. Since the message vectors are generated by neural networks and their length is set manually before training, we can set their length to a lower value to reduce entropy. In this investigative experiment, we test whether reducing message length is better than our scheme in terms of reducing entropy. Specifically, we test IC3NET-Original and IC3NET-DisEM with different message lengths in Predator Prey environments and exhibit results in Fig.\ref{fig:reducelength}. Each point represents a model's performance with its x-coordinate representing entropy and its y-coordinate representing timesteps. Therefore, points located in the bottom left signify good performance and low entropy. In PP-A, DisEM significantly reduces communication entropy with slight performance degradation. In PP-B, DisEM reduces communication entropy massively while maintaining baseline performance. In conclusion, even though reducing the message lengths of IC3NET-Original can reduce entropy, combining it with our scheme is better.

\begin{figure}[htbp]
\centering
\begin{minipage}[t]{0.23\textwidth}
\centering
\includegraphics[width=1\textwidth]{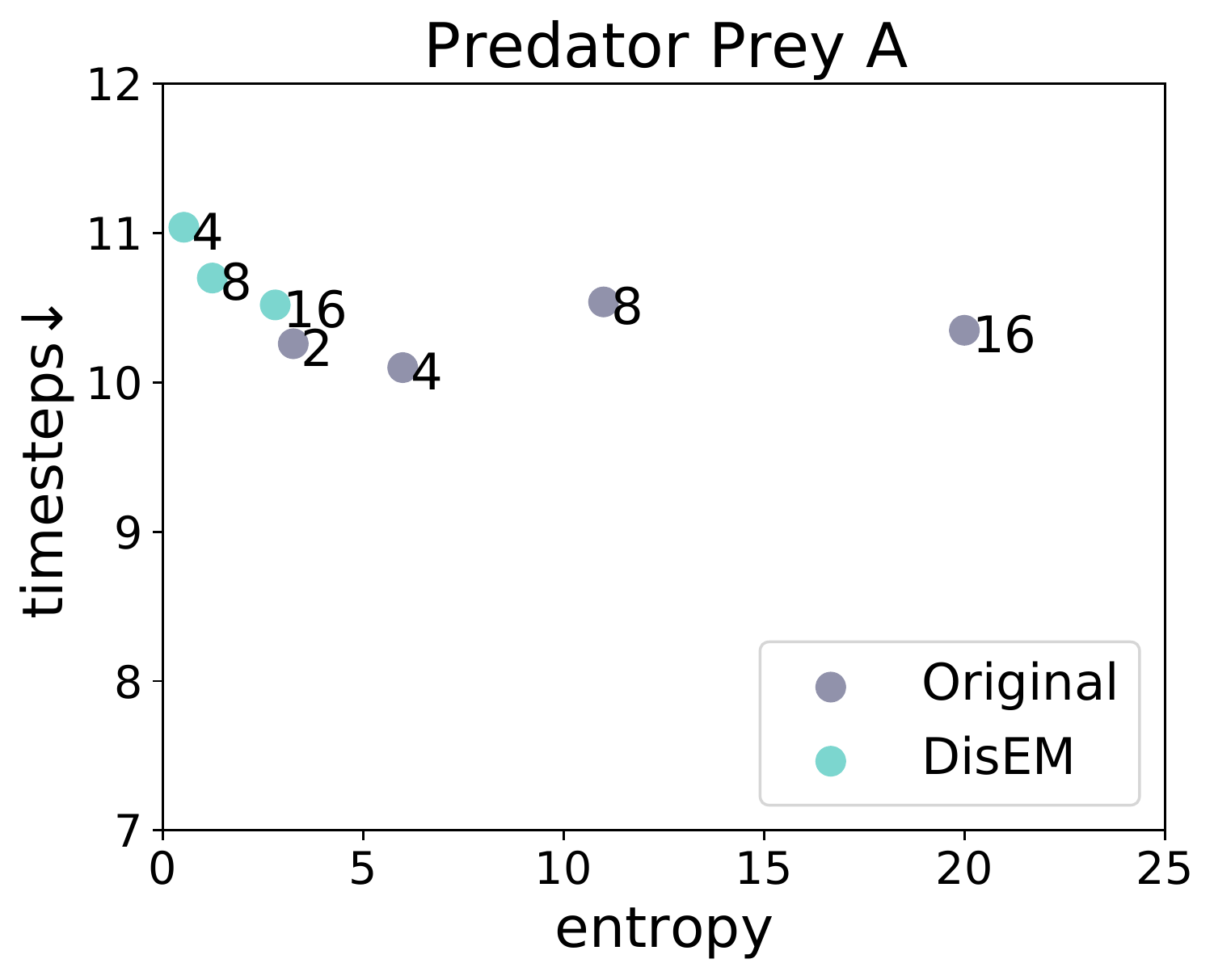}
\end{minipage}
\begin{minipage}[t]{0.23\textwidth}
\centering
\includegraphics[width=1\textwidth]{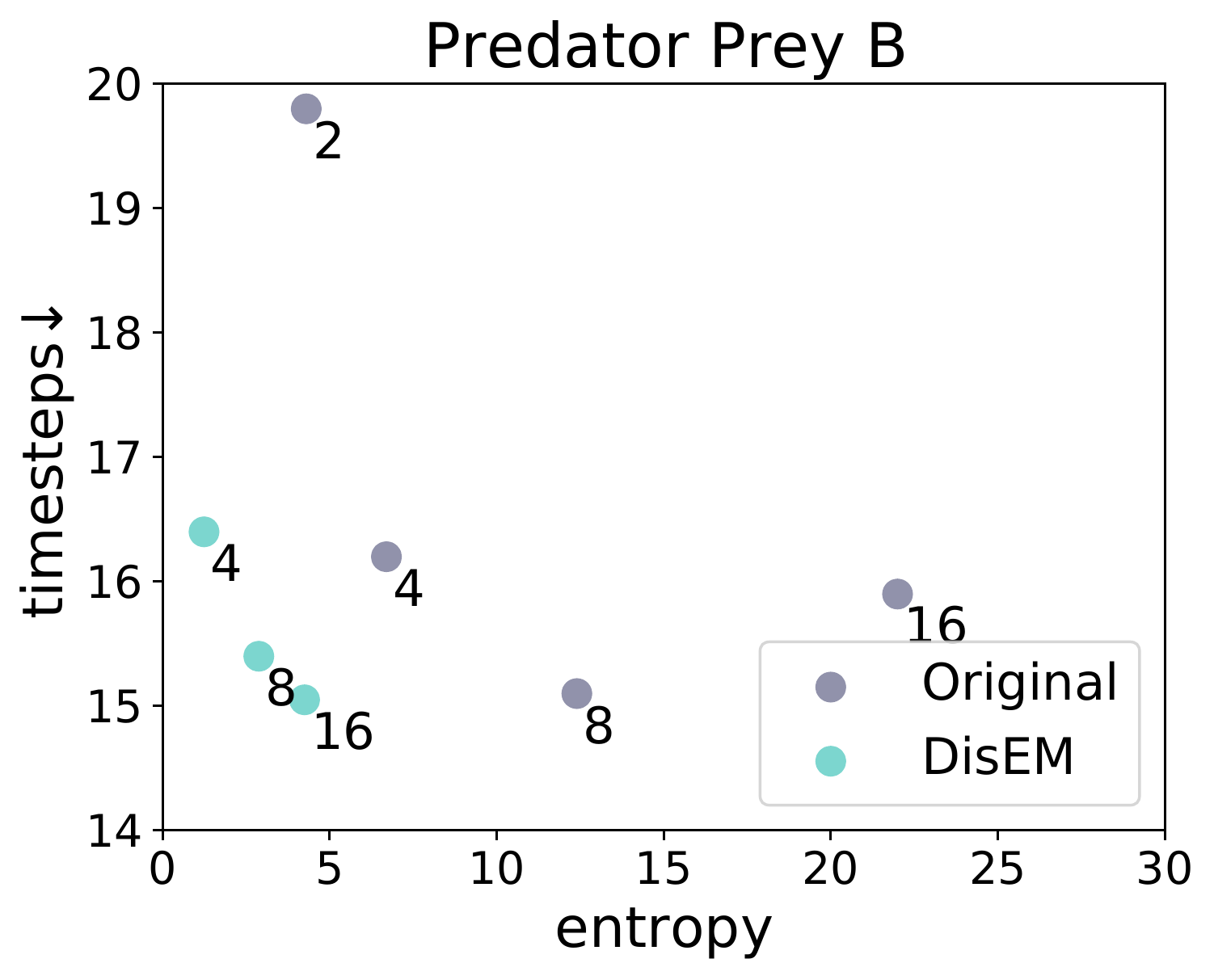}
\end{minipage}
\caption{Performance of IC3NET-Original and IC3NET-DisEM in Predator Prey environments. The number near a data point represents its message length, and the color of a data point indicates whether it is an IC3NET-Original model or an IC3NET-DisEM model. }
\label{fig:reducelength}
\end{figure}

\section{Conclusions}

In this paper, we propose a simple yet effective scheme, DisEM, to reduce communication entropy for common learning-based multi-agent communication frameworks. Firstly, we point out the necessity of minimizing the entropy of quantized messages in multi-agent communication. Secondly, to counter the problem that entropy cannot be optimized with gradient descent, we design pseudo gradient descent that reduces entropy by moving message variables from less popular quantization intervals to adjacent more popular ones. Thirdly, we prove the effectiveness of pseudo gradient descent theoretically. Fourthly, we conduct plenty of experiments to test our scheme. Concretely speaking, we test 8 variants of 2 base multi-agent communication frameworks, IC3NET and TARMAC, in six communication-critical environment settings. The results confirm the superiority of our scheme DisEM over the existing ones: DisEM can reduce the message entropy by up to 90$\%$ with nearly no loss of cooperation performance. We also conduct several investigative experiments and find that combining DisEM with reducing message sizes leads to lower message entropy. We hope our work can provide a foundation for more advanced research in efficient multi-agent communication.

\bibliographystyle{IEEEtran}
\bibliography{IEEEfull}
\vspace{12pt}

\end{document}